\documentclass[11pt,letterpaper]{article}

\usepackage{amsmath,amsfonts,amsthm,amssymb,stmaryrd,graphicx,relsize,bm}  
\theoremstyle{plain}

\numberwithin{equation}{section}

\newtheorem{thm}{Theorem}[section]
\newtheorem{lem}[thm]{Lemma}

\newenvironment{exam}[1]
{\begin{flushleft}\textbf{Example #1}.\enspace}%
{\end{flushleft}}

{\end{flushleft}}

\allowdisplaybreaks  

\newcounter{cond}

\newcommand{\real}{{\mathbb R}}

\newcommand{\trace}{tr}
\newcommand{\rmcor}{\mathrm{Cor}}
\newcommand{\instr}{In}
\newcommand{\ob}{Ob}
\newcommand{\ityes}{\textit{yes}} 
\newcommand{\itno}{\textit{no}} 
\newcommand{\rmre}{\mathrm{Re}}

\newcommand{\cscript}{\mathcal{C}}
\newcommand{\escript}{\mathcal{E}}
\newcommand{\hscript}{\mathcal{H}}
\newcommand{\iscript}{\mathcal{I}}
\newcommand{\jscript}{\mathcal{J}}
\newcommand{\kscript}{\mathcal{K}}
\newcommand{\lscript}{\mathcal{L}}
\newcommand{\oscript}{\mathcal{O}}
\newcommand{\sscript}{\mathcal{S}}

\newcommand{\iscripthat}{\widehat{\iscript}}
\newcommand{\jscripthat}{\widehat{\jscript}}
\newcommand{\kscripthat}{\widehat{\kscript}}
\newcommand{\atilde}{\widetilde{A}}
\newcommand{\btilde}{\widetilde{B}}
\newcommand{\ctilde}{\widetilde{C}}
\newcommand{\iscriptbar}{\overline{\iscript}}

\newcommand{\lscriptbar}{\overline{\lscript}}

\newcommand{\ab}[1]{\left|#1\right|}
\newcommand{\brac}[1]{\left\{#1\right\}}
\newcommand{\paren}[1]{\left(#1\right)}
\newcommand{\sqbrac}[1]{\left[#1\right]}
\newcommand{\elbows}[1]{{\left\langle#1\right\rangle}}
\newcommand{\ket}[1]{{\left|#1\right>}}
\newcommand{\bra}[1]{{\left<#1\right|}}

\begin{document}

\title{PROPERTIES OF SEQUENTIAL PRODUCTS}
\author{Stan Gudder\\ Department of Mathematics\\
University of Denver\\ Denver, Colorado 80208\\
sgudder@du.edu}
\date{}
\maketitle

\begin{abstract}
Our basic concept is the set $\escript (H)$ of effects on a finite dimensional complex Hilbert space $H$. If $a,b\in\escript (H)$, we define the sequential product 
$a\sqbrac{\iscript}b$ of $a$ then $b$. The sequential product depends on the operation $\iscript$ used to measure $a$. We begin by studying the properties of this sequential product. It is observed that $b\mapsto a\sqbrac{\iscript}b$ is an additive, convex morphism and we show by examples that $a\mapsto a\sqbrac{\iscript}b$ enjoys very few conditions. This is because a measurement of $a$ can interfere with a later measurement of $b$. We study sequential products relative to Kraus,
L\"uders and Holevo operations and find properties that characterize these operations. We consider repeatable effects and conditions on $a\sqbrac{\iscript}b$ that imply commutativity. We introduce the concept of an effect $b$ given an effect $a$ and study its properties. We next extend the sequential product to observables and instruments and develop statistical properties of real-valued observables. This is accomplished by employing corresponding stochastic operators. Finally, we introduce an uncertainty principle for conditioned observables.
\end{abstract}

\section{Introduction}  
In this work we consider quantum systems described by a finite dimensional complex Hilbert space $H$. Although this is a strong restriction, it is general enough to include quantum computation and information theory \cite{hz12,kw20,nc00}. In this introduction we give the reader general ideas behind the theory and leave detailed definitions for Section~2. Our basic concept is the set $\escript (H)$ of effects on $H$. An effect $a\in\escript (H)$ represents a two-valued
\ityes-\itno\ (or true-false) measurement. If the result of measuring $a$ is \ityes\ (true), we say that $a$ \textit{occurs} and otherwise $a$ \textit{does not occur}. The set of states on $H$ is denoted by $\sscript(H)$ where $\rho\in\sscript (H)$ provides the condition or state of knowledge of the system. If $a\in\escript (H)$ is measured and the system is in state $\rho\in\sscript (H)$, then the probability that outcome \ityes\ occurs is given by the trace $P_\rho (a)=\trace (\rho a)$. The effect that is always true is represented by the identity operator $I$ and $P_\rho (I)=1$ for all $\rho\in\sscript (H)$. Similarly, the effect that is always false is represented by the zero operator $0$ and $P_\rho (0)=0$ for all $\rho\in\sscript (H)$. An apparatus $\iscript ^a$ that is employed to measure an effect $a$ is called an
\textit{operation}. Suppose $a,b\in\escript (H)$ and $a$ is measured first using the operation $\iscript ^a$ and then $b$ is measured. The resulting effect is denoted
$a\sqbrac{\iscript ^a}b$ and is called the \textit{sequential product of} $a$ \textit{then} $b$ \textit{relative to} $\iscript ^a$. We see that the sequential product depends on the operation $\iscript ^a$ used to measure $a$ but does not depend on an operation used to measure $b$. Also, since $a$ was measured first, this can interfere with the $b$ measurement but not vice versa. It is important to point out that although an operation measures a unique effect, an effect can be measured by many operations. Moreover, as is expected in quantum mechanics, $a\sqbrac{\iscript ^a}b\ne b\sqbrac{\iscript ^b}a$, in general.

In Section~2, we discuss properties of the sequential product. For example, $0\sqbrac{\iscript ^0}a=a\sqbrac{\iscript ^a}0=0$ and $a\sqbrac{\iscript ^a}I=a$, but curiously, $I\sqbrac{\iscript ^I}a\ne a$ in general. We observe that $b\mapsto a\sqbrac{\iscript ^a}b$ is an additive, convex morphism and show by examples that, because $a$ can interfere with $b$, $a\mapsto a\sqbrac{\iscript ^a}b$ enjoys very few conditions. We say that $a\in\escript (H)$ is
\textit{repeatable relative to} $\iscript ^a$ if $a\sqbrac{\iscript ^a}a=a$. Characterizations of repeatability are given in Section~3. Conditions on
$a\sqbrac{\iscript ^a}b$ that imply commutativity are derived. Using the sequential product, we introduce the concept of an effect $b$ given an effect $a$ (or $b$
conditioned by $a$) which we denote by $b\ab{\iscript ^a}a$. Various properties of $b\ab{\iscript ^a}a$ are found. We define three types of operations called
Kraus, L\"uders and Holevo operations. Conditions that characterize sequential products relative to these operations are derived. Also, these operations are used to construct counterexamples for various properties of sequential products.

Section~4 extends the sequential product to observables and instruments. In fact, these are just collections of effects and operations, respectively, so their sequential product definitions are straightforward. We then study statistical properties of real-valued observables. In particular, if $A$ is a real-valued observable and
$\rho\in\sscript (H)$, we define the $\rho$-expectation and $\rho$-variance of $A$. If $A$ and $B$ are real-valued observables and $\rho\in\sscript (H)$, we define the $\rho$-correlation and $\rho$-covariance of $A,B$. These quantities are defined using the stochastic operators of the observables and are employed to form a generalized uncertainty principle.

\section{Sequential Products of Effects}  
We denote the set of bounded linear operators on $H$ by $\lscript (H)$. For $A,B\in\lscript (H)$ we write $A\le B$ if
$\elbows{\phi ,A\phi}\le\elbows{\phi ,B\phi}$ for all $\phi\in H$. We say that $A\in\lscript (H)$ is positive if $A\ge 0$ and $A$ is an effect if $0\le A\le I$
\cite{bgl95,dl70,hz12,kra83,nc00}. The set of effects on $H$ is denoted by $\escript (H)$. If $a\in\escript (H)$, then its complement $a'=I-$ is true if and only if $a$ is false. If $a,b\in\escript (H)$ and $a+b\in\escript (H)$ we write $a\perp b$ and interpret $a+b$ as the statistical sum of the measurements $a$ and $b$. Notice that
$a\perp b$ if and only if $b\le a'$. A map $K\colon\escript (H)\to\escript (H)$ is \textit{additive} if $K(a)\perp K(b)$ whenever $a\perp b$ and we then have
$K(a+b)=K(a)+K(b)$. If $K$ is additive, then $K$ preserves order because when $a\le b$, then there exists $a\in\escript (H)$ such that $a+c=b$ and we obtain
\begin{equation*}
K(a)\le K(a)+K(c)=K(a+c)=K(b)
\end{equation*}
If $K$ is additive and $K(I)=I$, then $K$ is a \textit{morphism} \cite{bgl95,hz12,nc00}.

A \textit{state} for a quantum system is a positive operator $\rho$ such that $\trace (\rho )=1$. We denote the set of states on $H$ by $\sscript (H)$. It is clear that for
$a,b\in\escript (H)$, we have $a\le b$ if and only if $P_\rho (a)\le P_\rho (b)$ and $P_\rho (a')=1-P_\rho (a)$ for all $\rho\in\sscript (H)$. An \textit{operation} on
$H$ is a completely positive linear map: $\iscript\colon\lscript (H)\to\lscript (H)$ that is trace non-increasing \cite{bgl95,dl70,hz12,kra83}. We denote the set of operations on $H$ by $\oscript (H)$. It can be shown \cite{bgl95,dl70,hz12,kra83} that every $\iscript\in\oscript (H)$ has a \textit{Kraus decomposition}
$\iscript (A)=\sum\limits _{i=1}^nK_iAK_i^*$, $A\in\lscript (H)$, where $K_i\in\lscript (H)$ satify $\sum\limits _{i=1}^nK_i^*K_i\le I$. The operators $K_i$ are called
\textit{Kraus operators} for $\iscript$. If an operation preserves the trace, it is called a \textit{channel} \cite{bgl95,hz12,kra83,nc00}. A \textit{dual operation} on $H$ is a completely positive linear map $\jscript\colon\lscript (H)\to\lscript (H)$ that satisfies $\jscript\colon\escript (H)\to\escript (H)$. We denote the set of dual operations on 
$H$ by $\oscript ^*(H)$. It is shown in \cite{bgl95,gud220,hz12} that if $\iscript\colon\lscript (H)\to\lscript (H)$ is an operation, then there exists an unique
$\iscript ^*\in\oscript ^*(H)$ such that 
\begin{equation*}
\trace\sqbrac{\rho\iscript ^*(a)}=\trace\sqbrac{\iscript (\rho )a}
\end{equation*}
for all $a\in\escript (H)$, $\rho\in\sscript (H)$. Moreover, $\iscript$ is a channel if and only if $\iscript ^*(I)=I$. In particular, if $\iscript$ has Kraus decomposition
$\iscript (A)=\sum K_iAK_i^*$, then its dual operation is $\iscript ^*(A)=\sum K_i^*AK_i$. We say that an operation $\iscript$ \textit{measures} an effect $a$ if
\begin{equation*}
\trace\sqbrac{\iscript (\rho )}=\trace (\rho a)=P_\rho (a)
\end{equation*}
for all $\rho\in\sscript (H)$. The operation $\iscript$ gives more information than the effect $a$. After $\iscript$ is performed, the state $\rho$ is updated to the state
$\iscript (\rho )/\trace\sqbrac{\iscript (\rho )}$ whenever, $\trace\sqbrac{\iscript (\rho )}\ne 0$. Although $\iscript$ measures a unique effect, as we shall see, an effect is measured by many operations. It is easy to show that the unique effect measured by $\iscript$ is $\iscripthat =\iscript ^*(I)$.

If $a,b\in\escript (H)$ and $\iscript ^a\in\oscript (H)$ measures $a$, we define the \textit{sequential product of} $a$ \textit{then} $b$ \textit{relative to} $\iscript ^a$ by
$a\sqbrac{\iscript ^a}b=\iscript ^{a*}(b)$ \cite{gg02,gn01,gud21,gud220}. We now discuss the properties of $a\sqbrac{\iscript ^a}b$. Since $\iscript ^{a*}$ is linear, we have
\begin{equation*}
a\sqbrac{\iscript ^a}0=\iscript ^{a*}(0)=0
\end{equation*}
If $\iscript ^0$ measures $0$, we have $\trace\sqbrac{\iscript ^0(\rho )}=\trace (\rho 0)=0$ for all $\rho\in\sscript (H)$. Hence, $\iscript ^0(\rho )=0$ for all
$\rho\in\sscript (H)$ and it follows that $\iscript ^{0*}(a)=0$ for all $a\in\escript (H)$. Hence,
\begin{equation*}
0\sqbrac{\iscript ^0}a=\iscript ^{0*}(a)=0
\end{equation*}
Also,
\begin{equation*}
a\sqbrac{\iscript ^a}I=\iscript ^{a*}(I)=a
\end{equation*}
for all $a\in\escript (H)$. We will show later that $I\sqbrac{\iscript ^I}a\ne a$, in general. If $b\perp c$, then 
\begin{equation*}
a\sqbrac{\iscript ^a}(b+c)=\iscript ^{a*}(b+c)=\iscript ^{a*}(b)+\iscript ^{a*}(c)=a\sqbrac{\iscript ^a}b+a\sqbrac{\iscript ^a}c
\end{equation*}
It follows that $b\mapsto a\sqbrac{\iscript ^a}b$ is a morphism on $\escript (H)$ for all $a\in\escript (H)$. Letting $c=b'$ we obtain
\begin{equation*}
a=a\sqbrac{\iscript ^a}I=a\sqbrac{\iscript ^a}(b+b')=a\sqbrac{\iscript ^a}b+a\sqbrac{\iscript ^a}b'
\end{equation*}
so that $a\sqbrac{\iscript ^a}b\le a$ for all $a,b\in\escript (H)$. If $b_i\in\escript (H)$, $\lambda _i\in\sqbrac{0,1}$, $i=1,2,\ldots ,n$ with $\sum\lambda _i=1$, then
$\sum\lambda _ib_i\in\escript (H)$ so $\escript (H)$ is a convex set. We then obtain 
\begin{equation*}
a\sqbrac{\iscript ^a}\paren{\sum\lambda _ib_i}=\iscript ^{a*}\paren{\sum\lambda _ib_i}=\sum\lambda _i\iscript ^{a*}(b_i)=\sum\lambda _ia\sqbrac{\iscript ^a}b_i
\end{equation*}
so $b\mapsto a\sqbrac{\iscript ^a}b$ is also a convex map on $\escript (H)$. We will show later that $a\mapsto a\sqbrac{\iscript ^a}b$ is not additive or convex, in general.

We now consider three types of operations. An operation $\iscript\in\oscript (H)$ is called a \textit{Kraus operation} if $\iscript (\rho )=K\rho K^*$ for all
$\rho\in\sscript (H)$, where $K\in\lscript (H)$ with $K^*K\le I$ \cite{kra83}. If $a=K^*K$, then $a\in\escript (H)$ and $\iscript ^*(I)=K^*K=a$ so
$\iscripthat =a$. By the polar decomposition theorem, there is a unitary operator $U$ on $H$ such that $K=a^{1/2}U$. Since
\begin{equation*}
a=K^*K=UaU^*
\end{equation*}
multiplying on the right with $U$ gives $aU=Ua$ so $a$ and $U$ commute. Since
\begin{equation*}
\trace\sqbrac{\rho\iscript ^*(b)}=\trace\sqbrac{\iscript (\rho )b}=\trace (a^{1/2}U\rho U^*a^{1/2}b)=\trace (\rho U^*a^{1/2}ba^{1/2}U)
\end{equation*}
we obtain
\begin{equation*}
a\sqbrac{\iscript}b=\iscript ^*(b)=U^*a^{1/2}ba^{1/2}U=a^{1/2}U^*bUa^{1/2}
\end{equation*}
We call $a\sqbrac{\iscript}b$ a \textit{Kraus sequential product} and sometimes write
\begin{equation}                
\label{eq21}
a\sqbrac{\iscript}b=a\sqbrac{\kscript ^a}b=K^*bK
\end{equation}

For $a\in\escript (H)$, we define the \textit{L\"uders operation} $\lscript ^a(\rho )=a^{1/2}\rho a^{1/2}$ \cite{lud51}. For $b\in\escript (H)$ we have
\begin{equation*}
\trace\sqbrac{\rho\lscript ^{a*}(b)}=\trace\sqbrac{\lscript ^a(\rho )b}=\trace (a^{1/2}\rho a^{1/2}b)=\trace (\rho a^{1/2}ba^{1/2})
\end{equation*}
Hence, $\lscript ^{a*}(b)=a^{1/2}ba^{1/2}$ and $(\lscript ^a)^\wedge =\lscript ^{a*}(I)=a$. We conclude that $\lscript ^a$ measures $a$. We see that $\lscript ^a$ is a specific kind of Kraus operation in which $K=a^{1/2}$ and the unitary operator $U=I$. The resulting sequential product becomes
\begin{equation*}
a\sqbrac{\lscript ^a}b=a^{1/2}ba^{1/2}
\end{equation*}
Unlike Kraus operators, $\lscript ^a$ is the only L\"uders operator that measures $a$. For this reason, we can omit the $a$ in $\lscript ^a$ and write the L\"uders sequential product in the simple form \cite{gg02,gn01,gud21}
\begin{equation}                
\label{eq22}
a\circ b=a\sqbrac{\lscript}b=a^{1/2}ba^{1/2}
\end{equation}

The third example of a operation is a \textit{Holevo  operation} $\hscript ^{(a,\alpha )}$ \cite{hol82}, $a\in\escript (H)$, $\alpha\in\sscript (H)$. We define
$\hscript ^{(a,\alpha )}(\rho )=\trace (\rho a)\alpha$ for all $\rho\in\sscript (H)$. Since
\begin{align*}
\trace\sqbrac{\rho\hscript ^{(a,\alpha )*}(b)}&=\trace\sqbrac{\hscript ^{(a,\alpha )}(\rho )b}=\trace\sqbrac{\trace (\rho a)\alpha b}=\trace (\rho a)\trace (\alpha b)\\
   &=\trace\sqbrac{\rho\,\trace (\alpha b)a}
\end{align*}
we obtain $\hscript ^{(a,\alpha )*}(b)=\trace (\alpha b)a$. Hence,
\begin{equation*}
\hscript ^{(a,\alpha )*}(I)=\trace (\alpha I)a=a
\end{equation*}
so $\hscript ^{(a,\alpha )}$ measures $a$. As the Kraus operations, this shows that there are many operations that measure an effect $a$. The resulting sequential product becomes
\begin{equation}                
\label{eq23}
a\sqbrac{\hscript ^{(a,\alpha )}}b=\trace (\alpha b)a
\end{equation}

If $\iscript ,\jscript$ measure $a,b$, respectively, so that $\iscript ^*(I)=a$, $\jscript ^*(I)=b$ we define the \textit{commutator of} $a,b$ \textit{relative to}
$\iscript ,\jscript$ by
\begin{equation}                
\label{eq24}
\cscript (a,b;\iscript ,\jscript )=\iscript ^*(b)-\jscript ^*(a)=a\sqbrac{\iscript}b-b\sqbrac{\jscript}a
\end{equation}
We say that $a,b$ \textit{commute relative to} $\iscript ,\jscript$ if $\cscript (a,b;\iscript ,\jscript )=0$. It can be shown \cite{gn01} that $a,b$ commute relative to
$\lscript ^a,\lscript ^b$ if and only if they commute in the usual sense; that is, $ab-ba=\sqbrac{a,b}=0$. For Kraus operations $\kscript ^a(\rho )=K\rho K^*$,
$\kscript ^b(\rho )=J\rho J^*$ we have
\begin{equation*}
\cscript (a,b;\kscript ^a,\kscript ^b)=K^*aK-J^*bJ
\end{equation*}
and for Holevo operations we obtain
\begin{equation*}
\cscript (a,b;\hscript ^{(a,\alpha )},\hscript ^{(b,\beta )})=\trace (\alpha b)a-\trace (\beta a)b
\end{equation*}

We now use the L\"uders sequential product \eqref{eq22} as a basis for comparing the properties of various sequential products. For simplicity if the L\"uders commutator vanishes, we write $\sqbrac{a,b}=0$. A list of the five basic properties of $a\circ b$ is:
\begin{list} {(L\arabic{cond})}{\usecounter{cond}
\setlength{\rightmargin}{\leftmargin}}
\item If $b\perp c$, then $a\circ (b+c)=a\circ b+a\circ c$ for all $a\in\escript (H)$ (additivity).
\item $I\circ a=a$ for all $a\in\escript (H)$.
\item If $a\circ b=0$, then $b\circ a=0$.
\item If $\sqbrac{a,b}=0$, then $\sqbrac{a,b'}=0$ and $a\circ (b\circ c)=(a\circ b)\circ c$ for all $c\in\escript (H)$.
\item If $\sqbrac{a,c}=0$ and $\sqbrac{b,c}=0$, then $\sqbrac{a\circ b,c}=0$ and $\sqbrac{a+b,c}=0$ where $a\perp b$.
\end{list}

\noindent It is shown in \cite{gg02} that these five properties hold. Conditions (L1)-(L5) essentially determine the L\"uders sequential product. In fact, if $a\circ b$ is a product on $\escript (H)$ satisfying these conditions, then $a\circ b =U^*a^{1/2}ba^{1/2}U$ where $U$ is a unitary operator depending on $a$ that commutes with $a$
\cite{wj09}.

We now consider the Holevo sequential product $a\sqbrac{\hscript ^{(a,\alpha )}}b=\trace (\alpha b)a$. The corresponding five conditions become:
\begin{list} {(H\arabic{cond})}{\usecounter{cond}
\setlength{\rightmargin}{\leftmargin}}
\item If $b\perp c$, then $a\sqbrac{\hscript ^{(a,\alpha )}}(b+c)=a\sqbrac{\hscript ^{(a,\alpha )}}b+a\sqbrac{\hscript ^{(a,\alpha )}}c$.
\item $I\sqbrac{\hscript ^{(I,\alpha )}}a=a$ for all $a\in\escript (H)$.
\item If $a\sqbrac{\hscript ^{(a,\alpha )}}b=0$, then $b\sqbrac{\hscript ^{(b,\alpha )}}a=0$.
\item If $\cscript (a,b;\hscript ^{(a,\alpha )},\hscript ^{(b,\alpha )})=0$, then $\cscript (a,b';\hscript ^{(a,\alpha )},\hscript ^{(b',\alpha )})=0$ and
\begin{equation}                
\label{eq25}
a\sqbrac{\hscript ^{(a,\alpha )}}\paren{b\sqbrac{\hscript ^{(b,\alpha )}}c}
    =\brac{a\sqbrac{\hscript ^{(a,\alpha )}}b}\sqbrac{\hscript ^{\paren{a\sqbrac{\hscript ^{(a,\alpha )}}b,\alpha}}}c
\end{equation}
\item If $\cscript (a,c;\hscript ^{(a,\alpha )},\hscript ^{(c,\alpha )})=0$ and $\cscript (b,c;\hscript ^{(b,\alpha )},\hscript ^{(c,\alpha )})=0$, then
\begin{align}                
\label{eq26}
\cscript&\paren{a\sqbrac{\hscript ^{(a,\alpha )}}b,c;\hscript ^{\paren{a\sqbrac{\hscript ^{(a,\alpha )}}b}},\hscript ^{(c,\alpha )}}=0\\
\intertext{and}
\label{eq27}
\cscript&(a+b,c;\hscript ^{(a+b,\alpha )},\hscript ^{(c,\alpha )})=0
\end{align}
\end{list}
Although (L1)-(L5) hold for the L\"uders product, the next result shows that some of the corresponding properties (H1)-(H5) hold for the Holevo product and some do not.

\begin{thm}    
\label{thm21}
Conditions (H1), (H5) hold and Conditions (H2), (H3) do not hold. For Condition (H4), if $\cscript (a,b;\hscript ^{(a,\alpha )},\hscript ^{(b,\alpha )})=0$, then we need not have $\cscript (a,b';\hscript ^{(a,\alpha )},\hscript ^{(b',\alpha )})=0$. However, \eqref{eq25} always holds.
\end{thm}
\begin{proof}
We have seen that all sequential products are additive so (H1) holds. To prove that (H5) holds, suppose $\cscript (a,c;\hscript ^{(a,\alpha )},\hscript ^{(b,\alpha )})=0$ and $\cscript (b,c;\hscript ^{(b,\alpha )},\hscript ^{(c,\alpha )})=0$. Then
\begin{align*}
\trace (\alpha c)a&=a\sqbrac{\hscript ^{(a,\alpha )}}c=c\sqbrac{\hscript ^{(c,\alpha )}}a=\trace (\alpha a)c\\
\intertext{and}
\trace (\alpha c)b&=b\sqbrac{\hscript ^{(b,\alpha )}}c=c\sqbrac{\hscript ^{(c,\alpha )}}b=\trace (\alpha b)c
\end{align*}
Hence,
\begin{align*}
\brac{a\sqbrac{\hscript ^{(a,\alpha )}}b}\sqbrac{\hscript ^{\paren{a\sqbrac{\hscript ^{(a,\alpha )}}b,\alpha}}}c
   &=\sqbrac{\trace (\alpha b)a}\sqbrac{\hscript ^{\paren{\trace (\alpha b),a,\alpha}}}c\\
   &=\trace (\alpha b)\trace (\alpha c)a=\trace (\alpha b)\trace (\alpha a)c\\
   &=\trace (\alpha b)c\sqbrac{\hscript ^{(c,\alpha )}}a=c\sqbrac{\hscript ^{(c,\alpha )}}\trace (\alpha b)a\\
   &=c\sqbrac{\hscript ^{(c,\alpha )}}\paren{a\sqbrac{\hscript ^{(a,\alpha )}}b}
\end{align*}
so \eqref{eq26} holds. Moreover,
\begin{align*}
(a+b)\sqbrac{\hscript ^{(a+b,\alpha )}}c&=\trace (\alpha c)(a+b)=\trace (\alpha c)a+\trace (\alpha c)b=\trace (\alpha a)c+\trace (\alpha b)c\\
    &=c\sqbrac{\hscript ^{(c,\alpha )}}a+c\sqbrac{\hscript ^{(c,\alpha )}}b=c\sqbrac{\hscript ^{(c,\alpha )}}(a+b)
\end{align*}
so \eqref{eq27} holds. To show that (H2) does not hold, we have
\begin{equation*}
I\sqbrac{\hscript ^{(I,\alpha )}}a=\trace (\alpha a)I
\end{equation*}
and $\trace (\alpha a)I\ne a$, in general. To show that (H3) does not hold, let $\phi ,\psi$ be orthogonal unit vectors in $H$ and let $\alpha =a=\ket{\phi}\bra{\phi}$ and 
$b=\elbows{\psi}\bra{\psi}$. Then
\begin{equation*}
a\sqbrac{\hscript ^{(a,\alpha )}}b=\trace (\alpha b)a=0
\end{equation*}
but
\begin{equation*}
b\sqbrac{\hscript ^{(b,\alpha )}}a=\trace (\alpha a)b=b\ne 0
\end{equation*}
For (H4), suppose $\cscript\paren{a,b;\hscript ^{(a,\alpha )},\hscript ^{(b,\alpha )}}=0$ and $a\ne\trace (\alpha a)I$. If\newline
$\cscript\paren{a,b';\hscript ^{(a,\alpha )},\hscript ^{(b',\alpha )}}=0$ we have $\trace (\alpha b)a=\trace (\alpha a)b$ and $\trace (\alpha b')a=\trace (\alpha a)b'$. Hence,
\begin{equation*}
a-\trace (\alpha b)a=\trace (\alpha a)I-\trace (\alpha a)b=\trace (\alpha a)I-\trace (\alpha b)a
\end{equation*}
This implies that $a=\trace (\alpha a)I$ which is a contradiction. We conclude that the first part of (H4) does not hold. To show that \eqref{eq25} always holds, we have
\begin{align*}
a\sqbrac{\hscript ^{(a,\alpha )}}\paren{b\sqbrac{\hscript ^{(b,\alpha )}}c}&=\trace (\alpha b)\trace (\alpha c)a=\trace (\alpha c)a\sqbrac{\hscript ^{(a,\alpha )}}b\\
   &=\brac{a\sqbrac{\hscript ^{(a,\alpha )}b}}\sqbrac{\hscript ^{\paren{a\sqbrac{\hscript ^{(a,\alpha )}}b,\alpha}}}c\qedhere
\end{align*}
\end{proof}

We now briefly consider the Kraus product $a\sqbrac{\kscript}b=\kscript ^*(b)=K^*bK$ where $K^*K=a$. The counterpart to (L1) holds as usual. To show that the counterpart to (L3) holds suppose $a\sqbrac{\kscript}b=0$. Since $K=a^{1/2}U$ where $U$ is unitary we obtain
\begin{equation*}
U^*a^{1/2}ba^{1/2}U=K^*bK=0
\end{equation*}
Hence, $a^{1/2}ba^{1/2}=0$ which implies $b^{1/2}ab^{1/2}=0$. Now $b\sqbrac{\jscript}a=J^*aJ$ where $J^*J=b$. We have that $J=b^{1/2}V$ where $V$ is unitary. Therefore,
\begin{equation*}
b\sqbrac{\jscript}a=V^*b^{1/2}ab^{1/2}V=0
\end{equation*}
so the counterpart to (L3) holds. We now show that the counterpart to (L2) need not hold. Suppose $\kscript$ is Kraus and $\kscript$ measures $I$ so $K^*K=I$. Now $K$ is unitary and if $a\ne 0,I$ there exist many unitary $K$ such that $K^*aK\ne a$. For such $K,I\sqbrac{\kscript}a\ne a$ so the counterpart of (L2) need not hold. We next show that the counterpart of (L4) does not hold, in general. For all $a\in\escript (H)$ we have $\cscript (a,0;\kscript ,\iscript )=0$. Now $0'=I$ and
$\cscript (a,I;\kscript ,\jscript )=0$ if and only if $a=a\sqbrac{\kscript}I=I\sqbrac{\jscript}a$. But since the counterpart of (L2) does not hold there exist an
$a\in\escript (H)$ such that $a\ne I\sqbrac{\jscript}a$. Hence, the counterpart of (L2) does not hold, in general. It can be shown that the counterpart of (L5) need not hold.

\section{Commutators, Repeatability and Conditioning of Effects}  
We first discuss commutators of effects. If $\lscript ^a,\lscript ^b$ are L\"uders operations, then
\begin{equation*}
\cscript (a,b;\lscript ^a,\lscript ^b)=a^{1/2}ba^{1/2}-b^{1/2}ab^{1/2}
\end{equation*}
In this case $\cscript (a,b;\lscript ^a,\lscript ^b)=0$ if and only if $\sqbrac{a,b}=0$. If $\hscript ^{(a,\alpha )},\hscript ^{(b,\beta )}$ are Holevo operations, then 
\begin{equation*}
\cscript (a,b;\hscript ^{(a,\alpha )},\hscript ^{(b,\beta )})=\hscript ^{(a,\alpha )*}(b)-\hscript ^{(b,\beta )*}(a)=\trace (\alpha b)a-\trace (\beta a)b
\end{equation*}
Hence, $\cscript (a,b;\hscript ^{(a,\alpha )},\hscript ^{(b,\beta )})=0$ if and only if $\trace (\alpha b)a=\trace (\beta a)b$. We conclude that 
$\cscript (a,b;\hscript ^{(a,\alpha )},\hscript ^{(b,\beta )})=0$ implies $\sqbrac{a,b}=0$ but the converse need not hold. Finally, if $\jscript (\rho )=J\rho J^*$,
$\kscript (\rho )=K\rho K^*$ are Kraus operations with $J^*J=a$ and $K^*K=b$, then
\begin{equation*}
\cscript (a,b;\jscript ,\kscript )=\jscript ^*(b)-\kscript ^*(a)=J^*bJ-K^*aK
\end{equation*}
Hence, $\cscript (a,b;\jscript ,\kscript )=0$ if and only if $J^*bJ=K^*aK$. Now $J=Ua^{1/2}$ and $K=Vb^{1/2}$ for unitary operators $U,V$. We conclude that
\begin{equation*}
U^*a^{1/2}ba^{1/2}U=V^*b^{1/2}ab^{1/2}V
\end{equation*}
so $a^{1/2}ba^{1/2}=(VU^*)^*b^{1/2}ab^{1/2}(VU^*)$. Thus, $\cscript (a,b;\jscript ,\kscript )=0$ if and only if $a^{1/2}ba^{1/2}$ is unitarily equivalent to
$b^{1/2}ab^{1/2}$.

\begin{exam}{1}  
If $a\le b$, then $c\sqbrac{\iscript}a\le c\sqbrac{\iscript}b$ because $b-a\in\escript (H)$ and
\begin{equation*}
c\sqbrac{\iscript}b-c\sqbrac{\iscript}a=c\sqbrac{\iscript}(b-a)\ge 0
\end{equation*}
However, if different operations are used, then $a\le b$ need not imply $c\sqbrac{\iscript}a\le c\sqbrac{\jscript}b$. In fact, we even have
\begin{equation*}
c\sqbrac{\hscript ^{(c,\alpha )}}a=\trace (\alpha a)c\not\le\trace (\beta a)c=c\sqbrac{\hscript ^{(c,\beta )}}a
\end{equation*}
The same argument applies to $a\le b$ but $a\sqbrac{\hscript ^{(a,\alpha )}}c\not\le b\sqbrac{\hscript ^{(b,\alpha )}}c$. We also have
\begin{equation*}
a\sqbrac{\hscript ^{(a,\alpha )}}c=\trace (\alpha c)a\ne\trace (\beta c)a=a\sqbrac{\hscript ^{(b,\beta )}}c
\end{equation*}
Moreover, there are examples of L\"uders products in which $a\le b$ but 
\begin{equation*}
a\sqbrac{\lscript ^a}c=a^{1/2}ca^{1/2}\not\le b^{1/2}cb^{1/2}=b\sqbrac{\lscript ^b}c\hskip 7pc\qedsymbol
\end{equation*}
\end{exam}

\begin{exam}{2}  
If $\cscript (c,a;\iscript ,\jscript )=\cscript (c,b;\iscript ,\kscript )=0$, where $a,b,c\ne 0$ does 
\begin{equation}                
\label{eq31}
\cscript\paren{c,a\sqbrac{\jscript}b;\iscript ,\lscript}=0
\end{equation}
If $\iscript ,\jscript ,\kscript ,\lscript$ are L\"uders, we have $\sqbrac{c,a}=\sqbrac{c,b}=0$ and it follows that 
\begin{equation*}
c^{1/2}a^{1/2}ba^{1/2}c^{1/2}=(a^{1/2}ba^{1/2})^{1/2}c(a^{1/2}ba^{1/2})^{1/2}
\end{equation*}
so \eqref{eq31} holds. What if $\iscript ,\jscript ,\kscript$ are L\"uders and $\lscript$ is Holevo $\hscript ^{(a^{1/2}ba^{1/2},\alpha )}$? We then obtain
$\sqbrac{c,a}=\sqbrac{c,b}=0$ and
\begin{align*}
\cscript\paren{c,a\sqbrac{\jscript}b;\iscript ,\lscript}&=\cscript (c,a^{1/2}ba^{1/2};\iscript ,\jscript )=c^{1/2}a^{1/2}ba^{1/2}c^{1/2}-\trace (\alpha c)a^{1/2}ba^{1/2}\\
   &=\sqbrac{c-\trace (\alpha c)I}a^{1/2}ba^{1/2}
\end{align*}
This vanishes if and only if $c=\trace (\alpha c)I$ or $a^{1/2}ba^{1/2}=0$ so \eqref{eq31} need not hold. Now suppose $\iscript ,\jscript ,\kscript ,\lscript$ are Holevo with states $\alpha ,\beta ,\gamma ,\delta$, respectively. Then $\trace (\alpha a)c=\trace (\beta c )a$ and $\trace (\alpha b)c=\trace (\gamma c)b$. We then obtain
\begin{align*}
\cscript\paren{c,a\sqbrac{\jscript}b;\iscript ,\lscript}&=\cscript\paren{c,\trace (\beta b)a;\iscript ,\jscript}
   =\trace (\beta b)\trace (\alpha a)c-\trace (\beta b)\trace (\delta c)a\\
   &=\trace (\beta b)\sqbrac{\trace (\beta c)-\trace (\delta c)}a=\trace (\beta b)\brac{\trace\sqbrac{(\beta -\delta )c}}a
\end{align*}
This vanishes if and only if $(\beta -\delta )c=0$ or $\beta b=0$ so \eqref{eq31} need not hold.\hfill\qedsymbol
\end{exam}

\begin{exam}{3}  
If $\cscript (c,a;\iscript ,\jscript )=\cscript (c,b;\iscript ,\kscript )=0$ where $a,b,c\ne 0$, does $\cscript (c,a+b;\iscript ,\lscript )=0$?
If $\iscript ,\jscript ,\kscript ,\lscript$ are L\"uders, the answer is clearly, yes. Let $\iscript ,\jscript ,\kscript ,\lscript$ be Holevo with states
$\alpha ,\beta ,\gamma ,\delta$, respectively. Then $\trace (\alpha a)c=\trace (\beta c)a$ and $\trace (\alpha b)c=\trace (\gamma c)b$. It follows that
\begin{align*}
\cscript (c,a+b;\iscript ,\lscript )&=\trace\sqbrac{\alpha (a+b)}c-\trace (\delta c)(a+b)\\
   &=\trace (\alpha a)c+\trace (\alpha b)c-\trace (\delta c)a-\trace (\delta c)b\\
   &=\trace (\beta  c)a+\trace (\gamma c)b-\trace (\delta c)a-\trace (\delta c)b\\
   &=\trace\sqbrac{(\beta -\delta )c}a-\trace\sqbrac{(\gamma -\delta )c}b
\end{align*}
In general, this is nonzero.\hfill\qedsymbol
\end{exam}

\begin{exam}{4}  
This example shows that Kraus products do not enjoy some of the properties enjoyed by L\"uders products. Let $a,b\in\escript (H)$ be given by
\begin{equation*}
a=\begin{bmatrix}1&0\\\noalign{\smallskip}0&1/2\end{bmatrix},\quad
b=\begin{bmatrix}0&0\\\noalign{\smallskip}0&1\end{bmatrix}
\end{equation*}
and let $\jscript ,\kscript$ be the Kraus operations $\jscript (\rho )=J\rho J^*$, $\kscript (\rho )=K\rho K^*$ where
\begin{equation*}
J=\begin{bmatrix}1&0\\\noalign{\smallskip}0&1/\sqrt{2}\end{bmatrix},\quad
K=\begin{bmatrix}0&1\\\noalign{\smallskip}0&0\end{bmatrix}
\end{equation*}
Since $\jscript ^*\jscript =a$ and $K^*K=b$, we have that $\jscript$ measures $a$ and $\kscript$ measures $b$. Now $\sqbrac{a,b}=0$ but
\begin{align*}
\cscript&(a,b;\jscript ,\kscript )=\jscript ^*(b)-\kscript ^*(a)=J^*bJ=K^*aK\\
   &=\begin{bmatrix}1&0\\\noalign{\smallskip}0&1/\sqrt{2}\end{bmatrix}\ \begin{bmatrix}0&0\\\noalign{\smallskip}0&1\end{bmatrix}\ 
    \begin{bmatrix}0&0\\\noalign{\smallskip}0&1/\sqrt{2}\end{bmatrix}-\begin{bmatrix}0&0\\\noalign{\smallskip}1&0\end{bmatrix}\ 
   \begin{bmatrix}1&0\\\noalign{\smallskip}0&1/2\end{bmatrix}\  \begin{bmatrix}0&1\\\noalign{\smallskip}0&0\end{bmatrix}\\
   &=\begin{bmatrix}0&0\\\noalign{\smallskip}0&-1/2\end{bmatrix}\ne 0\hskip 20pc\qedsymbol
\end{align*}
\end{exam}

An effect $a\in\escript (H)$ is \textit{repeatable with respect to} $\iscript\in\oscript (H)$ if $a\sqbrac{\iscript}a=a$ \cite{hz12}. Thus $a$ is repeatable with respect to
$\iscript$ if and only if $\iscript ^*(a)=a$ where $\iscripthat =a$. This is equivalent to $\trace (\rho a)=\trace\sqbrac{\iscript (\rho )a}$ for all $\rho\in\sscript (H)$. We say that $a$ is \textit{repeatable} if there exists an $\iscript\in\oscript (H)$ such that $a$ is repeatable with respect to $\iscript$.

\begin{thm}    
\label{thm31}
The following statements are equivalent.
{\rm{(i)}}\enspace $a$ is repeatable with respect to $\iscript$.
{\rm{(ii)}}\enspace $a\sqbrac{\iscript}a'=0$.
{\rm{(iii)}}\enspace $a\sqbrac{\iscript}b=0$ whenever $a\perp b$.
{\rm{(iv)}}\enspace $\iscript ^*(b)\le\iscript ^*(a)$ for all $b\in\escript (H)$.
{\rm{(v)}}\enspace $\trace\sqbrac{\iscript (\iscript (\rho ))}=\trace\sqbrac{\iscript (\rho )}$ for all $\rho\in\sscript (H)$.
\end{thm}
\begin{proof}
(i)$\Leftrightarrow$(ii)\enspace $a\sqbrac{\iscript}a'=0\Leftrightarrow\iscript ^*(a')=0\Leftrightarrow\iscript ^*(I)-\iscript (a)=0\Leftrightarrow\iscript ^*(a)=\iscript ^*(I)=a$.
(ii)$\Leftrightarrow$(iii)\enspace If $a\perp b$, then $b\le a'$ so $a\sqbrac{\iscript}b\le a\sqbrac{\iscript}a'=0$. Conversely, if $a\sqbrac{\iscript}b=0$ whenever
$a\perp b$, then since $a\perp a'$ we have $a\sqbrac{\iscript}a'=0$.
(i)$\Leftrightarrow$(iv)\enspace If $\iscript ^*(b)\le\iscript ^*(a)$ for all $b\in\escript (H)$, then $a=\iscript (I)\le\iscript ^*(a)$. Since $\iscript ^*(a)\le\iscript ^*(I)=a$,  we have $\iscript ^*(a)=a$. Conversely, if $\iscript ^*(a)=a$, then $\iscript ^*(b)\le a=\iscript ^*(a)$.
(i)$\Leftrightarrow$(v)\enspace If (v) holds, then for every $\rho\in\sscript (H)$ we obtain
\begin{align*}
\trace (\rho a)&=\trace\sqbrac{\rho\iscript ^*(I)}=\trace\sqbrac{\iscript (\rho )}=\trace\sqbrac{\iscript\paren{\iscript (\rho )}}=\trace\sqbrac{\iscript (\rho )\iscript ^*(I)}\\
   &=\trace\sqbrac{\iscript (\rho )a}=\trace\sqbrac{\rho\iscript ^*(a)}
\end{align*}
Hence, $\iscript ^*(a)=a$ so $a$ is repeatable with respect to $\iscript$. Conversely, if $a$ is repeatable with respect to $\iscript$, then for every $\rho\in\sscript (H)$ we have
\begin{align*}
\trace\sqbrac{\iscript (\rho )}&=\trace\sqbrac{\rho\iscript ^*(I)}=\trace (\rho a)=\trace\sqbrac{\rho\iscript ^*(a)}=\trace\sqbrac{\iscript (\rho )a}\\
    &=\trace\sqbrac{\iscript (\rho )\iscript ^*(I)}=\trace\sqbrac{\iscript\paren{\iscript (\rho )}}\qedhere
\end{align*}
\end{proof}

The next result is known \cite{hz12} but our proof is different.

\begin{thm}    
\label{thm32}
An effect $a$ is repeatable if and only if $a=0$ or has eigenvalue $1$.
\end{thm}
\begin{proof}
Notice that $0$ is repeatable. Suppose $a$ is repeatable with respect to $\iscript$ and $a\ne 0$. Then $\iscript ^*(I)=\iscript ^*(a)=a$. Let $\rho _0\in\sscript (H)$ satisfy
$\iscript (\rho _0)\ne 0$ and let $\rho _1=\iscript (\rho _0)/\trace\sqbrac{\iscript (\rho _0)}$. Then $\rho _1\in\sscript (H)$ and applying Theorem~\ref{thm31}(v) we have
\begin{equation*}
1=\frac{\trace\sqbrac{\iscript (\rho _0)}}{\trace\sqbrac{\iscript (\rho _0)}}=\frac{\trace\sqbrac{\iscript\paren{\iscript (\rho _0)}}}{\trace\sqbrac{\iscript (\rho _0)}}
   =\trace\sqbrac{\iscript (\rho _1)}=\trace\sqbrac{\rho _1\iscript ^*(I)}=\trace (\rho _1,a)
\end{equation*}
It follows that if $\psi$ is an eigenvector of $\rho _1$, then $a\psi =\psi$. Therefore, $a$ has eigenvalue $1$. Conversely, suppose $a\psi =\psi$ for a unit vector
$\psi\in H$ and let $\alpha =\ket{\psi}\bra{\psi}\in\sscript (H)$. We then obtain
\begin{equation*}
\hscript ^{(a,\alpha )*}(a)=\trace (\alpha a)a=\elbows{\psi ,a\psi}a=a
\end{equation*}
Therefore, $a$ is repeatable with respect to $\hscript ^{(a,\alpha )}$.
\end{proof}

We say that $a\in\escript (H)$ is \textit{sharp} if $a$ is a projection; that is $a^2=a$ \cite{hz12}.

\begin{thm}    
\label{thm33}
The following statements are equivalent.
{\rm{(i)}}\enspace $a$ is repeatable with respect to a Kraus operation if and only if $a$ is sharp.
{\rm{(ii)}}\enspace $a$ is repeatable with respect to a Holevo operation $\hscript ^{(a,\alpha )}$ if and only if $a=0$ or $a\psi =\psi$ for every eigenvector $\psi$ of
$\alpha$.
\end{thm}
\begin{proof}
(i)\enspace If $a$ is sharp, then $a$ is repeatable with respect to the Kraus operation $\kscript (\rho )=a^{1/2}\rho a^{1/2}$. Conversely, suppose $a$ is repeatable with respect to the Kraus operation $\kscript (\rho )=K\rho K^*$ where $K^*K=a$. Then $K=a^{1/2}U=Ua^{1/2}$ for a unitary operation $U$. Hence,
\begin{equation*}
a=\kscript ^*(a)=K^*aK=U^*a^{1/2}aa^{1/2}U=U^*a^2U=a^2
\end{equation*}
so $a$ is sharp.
(ii)\enspace $a$ is repeatable with respect to $\hscript ^{(a,\alpha )}$ if and only if 
\begin{equation*}
a=\hscript ^{(a,\alpha )*}(a)=\trace (\alpha a)a
\end{equation*}
It follows that $a=0$ or $\trace (\alpha a)=1$. The later condition is equivalent to $a\psi =\psi$ for every eigenvector $\psi$ of $\alpha$.
\end{proof}

It is trivial that an effect is repeatable with respect to a L\"uders operation if and only if it is sharp.

For effects $a,b$ and operations $\iscript ,\jscript$ that measure $a,a'$, respectively, we define $b$ \textit{conditioned by} $a$ \textit{relative to} $(\iscript ,\jscript )$ as
\begin{equation*}
b\ab{\iscript ,\jscript}a=a\sqbrac{\iscript}b+a'\sqbrac{\jscript}b=(\iscript ^*+\jscript ^*)(b)=(\iscript +\jscript )^*(b)
\end{equation*}
Notice that
\begin{equation*}
b\ab{\iscript ,\jscript}a+b'\ab{\iscript ,\jscript}a=\iscript ^*(I)+\jscript ^*(I)=a+a'=I
\end{equation*}
Hence, $b\ab{\iscript ,\jscript}a\in\escript (H)$ and $b'\ab{\iscript ,\jscript}a=(b\ab{\iscript ,\jscript}a)'$. We have that
\begin{equation*}
I\ab{\iscript ,\jscript}a=a\sqbrac{\iscript}I+a'\sqbrac{\jscript}I=a+a'=I
\end{equation*}
However, if $\iscript$ is an arbitrary channel, then $\iscript$ measures $I$ and the zero operation $\jscript (\rho )=0$ for every $\rho\in\sscript (H)$ measures $0=I'$. We then obtain
\begin{equation*}
b\ab{\iscript ,\jscript}I=\iscript ^*(b)+\jscript ^*(b)=\iscript ^*(b)
\end{equation*}
and $\iscript ^*(b)\ne b$, in general. It is clear that if $b\perp c$, then
\begin{equation*}
(b+c)\ab{\iscript ,\jscript}a=b\ab{\iscript ,\jscript}a+c\ab{\iscript ,\jscript}a
\end{equation*}
so $b\mapsto b\ab{\iscript ,\jscript}a$ is a morphism on $\escript (H)$. In a similar way, we have that $b\mapsto b\ab{\iscript ,\jscript}a$ is a convex morphism on
$\escript (H)$.

We now find $b\ab{\iscript ,\jscript}a$ for the various operations we considered. If $\iscript ,\jscript$ are Kraus operations, $\iscript (\rho )=K\rho K^*$,
$\jscript (\rho )=J\rho J^*$ where $\iscripthat =a$, $\jscripthat =a'$, then $K^*K=a$ and $J^*J=a'$. We then have unitary operators $U,V$ such that $K=Ua^{1/2}$ and $J=V(a')^{1/2}$. Hence, $\iscript ^*(b)=a^{1/2}U^*bUa^{1/2}$ and $\jscript ^*(b)=(a')^{1/2}V^*bV(a')^{1/2}$ so we obtain
\begin{equation*}
b\ab{\iscript ,\jscript}a=a^{1/2}U^*bUa^{1/2}+(a')^{1/2}V^*bV(a')^{1/2}
\end{equation*}
In particular, if $\iscript ,\jscript$ are L\"uders operations, then
\begin{equation*}
b\ab{\iscript ,\jscript}a=a^{1/2}ba^{1/2}+(a')^{1/2}b(a')^{1/2}
\end{equation*}
In the case of Holevo operations we have
\begin{align}                
\label{eq32}
b\ab{\hscript ^{(a,\alpha )},\hscript ^{(a',\beta )}}a&=\hscript ^{(a,\alpha )*}(b)+\hscript ^{(a',\beta )*}(b)=\trace (\alpha b)a+\trace (\beta b)a'\notag\\
   &=\trace\sqbrac{(\alpha -\beta )b}a+\trace (\beta b)I
\end{align}
In particular, when $\alpha =\beta$ we have
\begin{equation*}
b\ab{\hscript ^{(a,\alpha )},\hscript ^{(a',\alpha )}}a=\trace (\alpha b)I
\end{equation*}

If $b\ab{\iscript ,\jscript}a=b$, then $b$ is not affected by a previous measurement of $a$ using the operations $\iscript ,\jscript$. It is interesting that a measurement of $a$ can interfere with a later measurement of $a$. For example, although $a\ab{\iscript ,\jscript}a=a$ if $\iscript ,\jscript$ are Kraus or L\"uders operations, when
$\iscript ,\jscript$ are Holevo operations we obtain by \eqref{eq32}
\begin{equation*}
a\ab{\iscript ,\jscript}a=\trace\sqbrac{(\alpha -\beta )a}a+\trace (\beta a)I=\trace (\alpha a)a+\trace (\beta a)a'
\end{equation*}
In particular, if $\alpha =\beta$ we have $a\ab{\iscript ,\jscript}a=\trace (\alpha a)I$. The next result gives what we would expect for the special case when
$\iscript ,\jscript$ are L\"uders operations and $a$ is sharp.

\begin{thm}    
\label{thm34}
If $\sqbrac{a,b}=0$, then $b\ab{\lscript ^a,\lscript ^{a'}}a=b$. Conversely, if $a$ is sharp and $b\ab{\lscript ^a,\lscript ^{a'}}a=b$, then $\sqbrac{a,b}=0$.
\end{thm}
\begin{proof}
If $\sqbrac{a,b}=0$, then
\begin{equation*}
b\ab{\lscript ^a,\lscript ^{a'}}a=a^{1/2}ba^{1/2}+(a')^{1/2}b(a')^{1/2}=ba+b(I-a)=b
\end{equation*}
Conversely, suppose $a$ is sharp and $b\ab{\lscript ^a,\lscript ^{a'}}=b$. Then
\begin{equation*}
b=b\ab{\lscript ^a\lscript ^{a'}}a=aba+a'ba'
\end{equation*}
Multiplying by $a$ gives $ab=aba$ so
\begin{equation*}
ab=aba=(ab)^*=ba\qedhere
\end{equation*}
\end{proof}

\section{Sequential Products of Observables}  
This section extends the sequential product of effects to observables. A (finite) \textit{observable} is a set of effects $A=\brac{A_x\colon x\in\Omega _A}$ for which
$\sum\limits _{x\in\Omega _A}A_x=I$ \cite{gud120,gud21,gud220,hz12}. The finite set $\Omega _A$ is called the \textit{outcome space} for $A$. If a measurement of $A$ results in outcome $x\in\Omega _A$, we say that the effect $A_x$ \textit{occurs}. If $\rho\in\sscript (H)$, then $E_\rho (A_x)=\trace (\rho A_x)$ is the
\textit{probability that} $A_x$ \textit{occurs} when the system is in state $\rho$ and the \textit{probability distribution} of $A$ in state $\rho$ is the function of $x$ given by
$\Phi _\rho ^A(x)=\trace (\rho A_x)$. Note that $\Phi _\rho ^A$ is indeed a probability distribution because
\begin{equation*}
\sum _{x\in\Omega _A}\Phi _\rho ^A(x)=\sum _{x\in\Omega _A}\trace (\rho A_x)=\trace\paren{\rho\sum _{x\in\Omega _A}A_x}=\trace (\rho I)=\trace (\rho )=1
\end{equation*}
We denote the set of observables on $H$ by $\ob (H)$ and if $A\in\ob (H)$, $\Delta\subseteq\Omega _A$, we write $A(\Delta )=\sum\brac{A_x\colon x\in\Delta}$ and
$\Phi _\rho ^A(\Delta )=\sum\brac{\Phi _\rho ^A(x)\colon x\in\Delta}$.

A (finite) \textit{instrument} is a finite set of operations $\iscript =\brac{\iscript _x\colon x\in\Omega _\iscript}$ for which
$\iscriptbar =\iscript (\Omega _\iscript )=\sum\limits _{x\in\Omega _\iscript}\iscript _x$ is a channel \cite{hz12}. We denote the set of instruments on $H$ by
$\instr (H)$. The \textit{probability distribution} of $\iscript$ in the state $\rho$ is $\Phi _\rho ^\iscript (x)=\trace\sqbrac{\iscript _x(\rho )}$. Again, $\Phi _\rho$ is a probability distribution because 
\begin{equation*}
\sum _{x\in\Omega _\iscript}\Phi _\rho ^\iscript (x)=\sum _{x\in\Omega _\iscript}\trace\sqbrac{\iscript _x(\rho )}
   =\trace\sqbrac{\sum _{x\in\iscript _\Omega}\iscript _x(\rho )}=\trace\sqbrac{\,\iscriptbar (\rho )}=1
\end{equation*}
We say that $\iscript$ \textit{measures} the observable $A$ if $\trace (\rho A_x)=\trace\sqbrac{\iscript _x(\rho )}$ for all $x\in\Omega _A$, $\rho\in\sscript (H)$ and we then write $\iscripthat =A$. If $\iscripthat =A$, we think of $\iscript$ as an apparatus that can be employed to measure $A$. It is clear that an instrument $\iscript$ measures a unique observable $\iscripthat$. However, as we shall see, an observable is measured by many instruments. Notice that $\iscript$ measure $A$ if and only if $\trace (\rho A_x)=\trace\sqbrac{\rho\iscript _x^*(I)}$ for all $\rho\in\sscript (H)$ which is equivalent to $\iscript _x^*(I)=A_x$ for all
$x\in\Omega _a=\Omega _\iscript$.

A \textit{bi-observable} is an observable of the form
\begin{equation*}
A=\brac{A_{(x,y)}\colon (x,y)\in\Omega _1\times\Omega _2}
\end{equation*}
If $A$ is a bi-observable, we can form the two \textit{marginal observables}
\begin{align*}
A_x^1&=A_{(x,\Omega _2)}=\sum _{y\in\Omega _2}A_{(x,y)}\\
A_y^2&=A_{(\Omega _1,y)}=\sum _{x\in\Omega _1}A_{(x,y)}
\end{align*}
We say that $B,C\in\oscript (H)$ \textit{coexist} \cite{gud120,gud21,gud220,hz12} if there is a bi-observable $A$ such that $B_x=A_x^1$, $C_y=A_y^2$ for all 
$x\in\Omega _B$, $y\in\Omega _C$. If $\iscripthat =A$ and $B\in\oscript (H)$ is another observable, we define the bi-observable 
\begin{equation*}
A\sqbrac{\iscript}B_{(x,y)}=\iscript _x^*(B_y)=A_x\sqbrac{\iscript _x}B_y
\end{equation*}
We interpret $A\sqbrac{\iscript}B$ as the observable obtained by first measuring $A$ with the instrument $\iscript$ and then measuring $B$ and call
$A\sqbrac{\iscript}B$ the \textit{sequential product of} $A$ \textit{then} $B$ \textit{relative to} $\iscript$ \cite{gud120,gud21,gud220,hz12}. The marginal observables become
\begin{align*}
A\sqbrac{\iscript}B_{(x,\Omega _B)}&=\sum _{y\in\Omega _B}A\sqbrac{\iscript}B_{(x,y)}=\sum _{y\in\Omega _B}\iscript _x^*(B_y)
    =\iscript _x^*\paren{\sum _{y\in\Omega _B}B_y}\\
    &=\iscript _x^*(I)=A_x\\
    A\sqbrac{\iscript}B_{(\Omega _A,y)}&=\sum _{x\in\Omega _A}A\sqbrac{\iscript}B_{(x,y)}=\sum _{x\in\Omega _A}\iscript _x^*(B_y)=\iscript _{\Omega _A}^*(B_y)\\
    &=\iscriptbar^{\,*}(B_y)=\sum _{x\in\Omega _A}A_x\sqbrac{\iscript _x}B_y
\end{align*}
We define the observable $B$ \textit{conditioned on} $A$ \textit{relative to} $\iscript$ by
\begin{equation*}
\paren{B\ab{\iscript}A}_y=\sum _{x\in\Omega _A}A_x\sqbrac{\iscript _x}B_y
\end{equation*}
It follows that $\paren{B\ab{\iscript}A}$ and $A$ coexist with \textit{joint observable}
\begin{equation*}
\paren{B\ab{\iscript}A}_{(x,y)}=A_x\sqbrac{\iscript _x}B_y
\end{equation*}
A \textit{Kraus instrument} \cite{kra83} $\kscript$ has the form $\kscript _x(\rho )=K_x\rho K_x^*$ where $\sum\limits _xK_x^*K_x=I$. Then
$\kscript _x^*(a)=K_x^*aK_x$ for all $a\in\escript (H)$ and $\kscripthat _x=\kscripthat _x^{\,*}(I)=K_x^*K_x$ so $\kscript$ measures the observable $A_x=K_x^*K$. If $B\in\oscript (H)$, the sequential product becomes the bi-observable
\begin{equation*}
A\sqbrac{\kscript}B_{(x,y)}=\kscript _x^*(B_y)=K_x^*B_yK_x
\end{equation*}
Also, the observable $B$ conditioned by $A$ relative to $\kscript$ is
\begin{equation*}
\paren{B\ab{\kscript}A}_y=\sum _xA_x\sqbrac{\kscript _y}B_y=\sum _xK_x^*B_yK_x
\end{equation*}

If $A\in\ob (H)$, the corresponding \textit{L\"uders instrument} \cite{lud51} is $\lscript _x^A(\rho )=A_x^{1/2}\rho A_x^{1/2}$. Then
$\lscript _x^*(a)=A_x^{1/2}aA_x^{1/2}$ for all $a\in\escript (H)$ and $\lscript _x^{A^*}(I)=A_x$ so $\lscript ^A$ measures $A$. The sequential product bi-observable becomes
\begin{equation*}
A\sqbrac{\lscript ^A}B_{(x,y)}=A_x^{1/2}B_yA_x^{1/2}
\end{equation*}
and we have
\begin{equation*}
\paren{B\ab{\lscript ^A}A}_y=\sum _{x\in\Omega _A}A_x^{1/2}B_yA_x^{1/2}
\end{equation*}

For an observable $A=\brac{A_x\colon x\in\Omega _A}$ and states $\alpha _x\in\sscript (H)$ we define the \textit{Holevo instrument} \cite{hol82}
$\hscript _x^{(A,\alpha )}(\rho )=\trace (\rho A_x)\alpha _x$. Then
\begin{equation*}
\hscript _x^{(A,\alpha )*}(a)=\trace (\alpha _xa)A_x
\end{equation*}
for all $a\in\escript (H)$ so $\hscript _x^{(A,\alpha )*}(I)=A_x$ and $\hscript ^{(A,\alpha )}$ measures $A$. The sequential product bi-observable becomes
\begin{equation*}
A\sqbrac{\hscript ^{(A,\alpha )}}B_{(x,y)}=\hscript _x^{(A,\alpha )*}(B_y)=\trace (\alpha _xB_y)A_x
\end{equation*}
and we have
\begin{equation}                
\label{eq41}
\paren{B\ab{\hscript ^{(A,\alpha )}}A}_y=\sum _x\trace (\alpha _xB_y)A_x
\end{equation}
We call $\trace (\alpha _xB_y)$ a \textit{transition probability} because $0\le\trace (\alpha _xB_y)\le 1$ and $\sum\limits _y\trace (\alpha _xB_y)=1$ for all
$x\in\Omega _A$. The observable $\paren{B\ab{\hscript ^{(A,\alpha )}}A}_y$ is also called a \textit{post-processing} of the observable $A$ \cite{hz12}. An observable $A$ is \textit{sharp} if all its effects $A_x$ are sharp. If $A$ is sharp, then from \eqref{eq41} we have that $\paren{B\ab{\hscript ^{(A,\alpha )}}A}$ and
$\paren{\cscript\ab{\hscript ^{(A,\alpha )}}A}$ have commuting effects and hence coexist.

If $S\in\lscript _S(H)$, $\rho\in\sscript (H)$, we define the $\rho$-\textit{expectation} of $S$ by $E_\rho (S)=\trace (\rho S)$. For $S,T\in\lscript _S(H)$,
$\rho\in\sscript (H)$, we define the $\rho$-\textit{correlation} of $S,T$ as 
\begin{equation*}
\rmcor _\rho (S,T)=\trace (\rho ST)-E_\rho (S)E_\rho (T)=\trace (\rho ST)-\trace (\rho S)\trace (\rho T)
\end{equation*}
the $\rho$-\textit{covariance} of $S,T$ as
\begin{equation*}
\Delta _\rho (S,T)=\rmre\rmcor _\rho (S,T)=\rmre\,\trace (\rho ST)-\trace (\rho S)\trace (\rho T)
\end{equation*}
and the $\rho$-\textit{variance} of $S$ as
\begin{equation*}
\Delta _\rho (S)=\Delta _\rho (S,S)=\trace (\rho S^2)-\sqbrac{\trace (\rho S)}^2
\end{equation*}
The following theorem is called the \textit{uncertainty principle} \cite{gud23}.

\begin{thm}    
\label{thm41}
If $S,T\in\lscript _S(H)$, $\rho\in\sscript (H)$, then
\begin{equation*}
\ab{\trace\paren{\rho\sqbrac{S,T}}}^2+\sqbrac{\Delta _\rho (S,T)}^2=\ab{\rmcor _\rho (S,T)}^2\le\Delta _\rho (S)\Delta _\rho (T)
\end{equation*}
\end{thm}

In this theorem, $\Delta _\rho (S)\Delta _\rho (T)$ is interpreted as the product of the uncertainties in measuring $S$ and $T$ in the state $\rho$. The first term is called the \textit{commutator term} and the second is the \textit{covariance} term.

We now discuss the statistics of \textit{real-valued observables} which are observables of the form $A=\brac{A_x\colon x\in\Omega _A}$ where
$\Omega _A\subseteq\real$ \cite{gud23}. If $A$ is real-valued, we call the self-adjoint operator $\atilde =\sum xA_x$ the \textit{stochastic operator} for $A$. For
real-valued $A,B\in\ob (H)$, $\rho\in\sscript (H)$, we define $E_\rho (A)=E_\rho (\atilde\,)$,
$\rmcor _\rho (A,B)=\rmcor _\rho (\atilde ,\btilde\,)=\Delta _\rho (\atilde ,\btilde\,)$ and $\Delta _\rho (A)=\Delta _\rho (\atilde\,)$. If $A,B\in\ob (H)$ where $B$ is
real-valued and $\iscript\in\instr (H)$ measures $A$, then $B\ab{\iscript}A$ is real-valued and we have
\begin{equation*}
\paren{B\ab{\iscript}A}^\sim =\sum _yy\paren{B\ab{\iscript}A}_y=\sum _yy\iscriptbar ^{\,*}(B_y)=\iscriptbar ^{\,*}(\btilde\,)
\end{equation*}
The $\rho$-expectation becomes 
\begin{align*}
E_\rho\paren{B\ab{\iscript}A}&=E_\rho\sqbrac{\paren{B\ab{\iscript}A}^\sim}=\trace\sqbrac{\rho\paren{B\ab{\iscript}A}^\sim}
    =\trace\sqbrac{\rho\sum _yy\paren{B\ab{\iscript}A}_y}\\
    &=\trace\sqbrac{\rho\sum _yy\sum _x\iscript _x^*(B_y)}=\sum _{x,y}y\,\trace\sqbrac{\rho\iscript _x^*(B_y)}\\
    &=\trace\sqbrac{\rho\iscriptbar ^{\,*}(\btilde\,)}=\trace\sqbrac{\iscriptbar (\rho )\btilde\,}
\end{align*}
We also have
\begin{align*}
\Delta _\rho\paren{B\ab{\iscript}A}&=\trace\sqbrac{\rho\paren{B\ab{\iscript}A}^{\sim 2}}-\sqbrac{\trace\paren{\rho\paren{B\ab{\iscript}A}^\sim}}^2\\
   &=\trace\sqbrac{\rho\paren{\iscriptbar ^{\,*}(\btilde\,)}^2}-\trace\sqbrac{\paren{\iscriptbar (\rho )\btilde\,}}^2
\end{align*}
Moreover, if $C\in\ob (H)$ is real-valued, then $\paren{C\ab{\iscript}A}^\sim =\iscriptbar ^{\,*}(\ctilde\,)$ and we obtain
\begin{align*}
\rmcor _\rho\sqbrac{B\ab{\iscript}A,C\ab{\iscript}A}&=\rmcor _\rho\sqbrac{\paren{B\ab{\iscript}A}^\sim ,\paren{C\ab{\iscript}A}^\sim}\\
   &=\trace\sqbrac{\rho\paren{B\ab{\iscript}A}^\sim\paren{C\ab{\iscript}A}^\sim}
   -\trace\sqbrac{\rho\paren{B\ab{\iscript}A}^\sim}\trace\sqbrac{\rho\paren{C\ab{\iscript}A}^\sim}\\
   &=\trace\sqbrac{\rho\iscriptbar ^{\,*}(\btilde\,)\iscriptbar ^{\,*}(\ctilde )}-\trace\sqbrac{\rho\iscriptbar ^{\,*}(\btilde )}\trace\sqbrac{\rho\iscriptbar ^{\,*}(\ctilde )}
\end{align*}
The commutator term becomes
\begin{equation*}
\ab{\trace\paren{\rho\sqbrac{\paren{B\ab{\iscript}A}^\sim ,\paren{C\ab{\iscript}A}^\sim}}}^2
   =\ab{\trace\paren{\rho\sqbrac{\iscriptbar ^{\,*}(\btilde\,),\iscriptbar ^{\,*}(\ctilde\,)}}}^2
\end{equation*}
If we place these terms in Theorem~\ref{thm41} we obtin the uncertainty principle for $B\ab{\iscript}A$, $C\ab{\iscript}A$.

\begin{exam}{5}  
If $\lscript ^A$ is the L\"uders instrument that measures $A\in\ob (H)$ and $B\in\ob (H)$ is real-valued, then
\begin{align*}
\paren{B\ab{\lscript ^A}A}^\sim&=\sum _yy\paren{B\ab{\lscript ^A}A}_y=\sum _yy\sum _xA_x^{1/2}B_yA_x^{1/2}\\
   &=\sum _xA_x^{1/2}\btilde A_x^{1/2}
   =\lscriptbar ^{A*}(\btilde )
\end{align*}
We then obtain
\begin{align*}
E_\rho\paren{B\ab{\lscript ^A}A}&=E_\rho\sqbrac{\paren{B\ab{\lscript ^A}A}^\sim}=\trace\sqbrac{\rho\lscriptbar ^{A*}(\btilde\,)}
   =\trace\sqbrac{\lscriptbar ^A(\rho )\btilde}\\
   &\trace\paren{\sum _xA_x^{1/2}\rho A_x^{1/2}\btilde\,}=\sum _x\trace (A_x^{1/2}\rho A_x^{1/2}\btilde\,)\\
   \Delta _\rho\paren{B\ab{\lscript ^A}A}&=\trace\sqbrac{\rho\paren{\lscriptbar ^{A*}(\btilde\,)}^2}-\brac{\trace\sqbrac{\lscriptbar ^A(\rho )\btilde\,}}^2\\
   &=\trace\sqbrac{\rho\paren{\sum _xA_x^{1/2}\btilde A_x^{1/2}}^2}-\sqbrac{\sum _x\trace\paren{A_x^{1/2}\rho A_z^{1/2}\btilde\,}}^2
\end{align*}
If $A$ is sharp, the commutator reduces to
\begin{align*}
\sqbrac{\paren{B\ab{\lscript ^A}A}^\sim ,\paren{C\ab{\lscript ^A}A}^\sim}&=\sum _xA_x\btilde A_x\ctilde A_x-\sum _xA_x\ctilde A_x\btilde A_x\\
   &=\sum _xA_x\sqbrac{\btilde A_x,\ctilde A_x}A_x\hskip 5pc\qedsymbol
\end{align*}
\end{exam}

\begin{exam}{6}  
If $\hscript ^{(A,\alpha )}$ is Holevo and $B$ is real-valued, then 
\begin{equation*}
\paren{B\ab{\hscript ^{(A,\alpha )}}A}^\sim =\sum _xA_x\sqbrac{\hscript ^{(A,\alpha )}}\btilde =\sum _x\hscript _x^{(A,\alpha )*}(\btilde\,)
   =\sum _x\trace (\alpha _x\btilde\,)A_x
\end{equation*}
Therefore,
\begin{align*}
E_\rho\paren{B\ab{\hscript ^{(A,\alpha )}}A}&=\sum _x\trace (\rho A_x)\trace (\alpha _x\btilde\,)\\
   \Delta _\rho\paren{B\ab{\hscript ^{(A,\alpha )}}A}&=\trace\brac{\rho\sqbrac{\sum _x\trace (\alpha _x\btilde\,)A_x}^2}
   -\brac{\sum _x\trace (\rho A_x)\trace (\alpha _x\btilde\,)}^2
\end{align*}
If $\alpha _x=\alpha$ for all $x\in\Omega _A$, these last two equations reduce to
\begin{align*}
E_\rho\paren{B\ab{\hscript ^{(A,\alpha )}}A}&=\trace (\alpha\btilde\,)=E_\alpha (B)\\
   \Delta _\rho\paren{B\ab{\hscript ^{(A,\alpha )}}A}&=\trace\sqbrac{(\alpha\btilde\,)^2}-\sqbrac{\trace (\alpha\btilde\,)}^2
   =\trace\sqbrac{(\alpha\btilde\,)^2}-\sqbrac{E_\alpha (B)}^2\hskip 1pc\qedsymbol
\end{align*}
\end{exam}

\end{document}